\documentclass[letterpaper, 10 pt, conference]{ieeeconf}

\IEEEoverridecommandlockouts 
\title{\LARGE \bf Target Controllability and Target Observability \\ of Structured Network Systems}

\author{Arthur N. Montanari, Chao Duan, \textit{Member, IEEE}, and Adilson E. Motter, \textit{Senior Member, IEEE}
\thanks{The authors acknowledge support from US Army Research Office Grant W911NF-19-1-0383 (A.N.M. and A.E.M.) and the National Natural Science Foundation of China Grant GQQNKP001 (C.D.).}
\thanks{A. N. Montanari is with the Department of Physics and Astronomy, Northwestern University, Evanston, IL 60208, USA (e-mail: arthur.monta nari@northwestern.edu).
C.~Duan is with the School of Electrical Engineering, Xi'an Jiaotong University, Xi'an 710049, China. %(e-mail: cduan@xjtu.edu.cn).
A. E. Motter is with the Department of Physics and Astronomy, the Department of Engineering Sciences and Applied Mathematics, and the Northwestern Institute on Complex Systems, Northwestern University, Evanston, IL 60208, USA.} %(e-mail: motter@northwestern.edu).}
}

\usepackage[noadjust]{cite} 		  % citations
\usepackage{xcolor}       % colors
\usepackage{hyperref}	  % links and bookmarks
\hypersetup{colorlinks=true,linkcolor=blue,citecolor=blue,breaklinks={true}} %hyperref configuration
\usepackage{comment}

\usepackage{graphicx}
\DeclareGraphicsExtensions{.eps,.png,.jpg,.jpeg,.bmp,.gif,.pdf}

\usepackage{tabularx,booktabs}	% use more thick toplines and bottom lines
\usepackage{multirow}	% multiple columns in one cell
\usepackage{amsmath}      		% equations
\usepackage{amssymb}	  		% symbols
\usepackage{bm}		      		% bold symbols
\usepackage{physics}	  		% derivatives
\usepackage{soul}				% strikethrough command (\st)
\newcommand{\R}{\mathbb{R}}		% real numbers
\newcommand{\transp}{\mathsf{T}}					% transpose
\usepackage{amsthm}
\newtheorem{defin}{Definition}
\newtheorem{thm}{Theorem}
\newtheorem{cor}{Corollary}

\theoremstyle{definition}
\newtheorem{assump}{Assumption}
\newtheorem{rem}{Remark}
\newtheorem{example}{Example}

\renewcommand{\QED}{\QEDopen}

\newcommand*{\QEDwhite}{\hfill\ensuremath{\square}}
\begin{document}

\maketitle
\thispagestyle{empty}
\pagestyle{empty}

\begin{abstract}      
The duality between controllability and observability enables methods developed for full-state control to be applied to full-state estimation, and vice versa. In applications in which control or estimation of all state variables is unfeasible, the generalized notions of output controllability and functional observability establish the minimal conditions for the control and estimation of a target subset of state variables, respectively. Given the seemly unrelated nature of these properties, thus far methods for target control and target estimation have been developed independently in the literature. Here, we characterize the graph-theoretic conditions for target controllability and target observability (which are, respectively, special cases of output controllability and functional observability for structured systems). This allow us to rigorously establish a weak and strong duality between these generalized properties. When both properties are equivalent (strongly dual), we show that efficient algorithms developed for target controllability can be used for target observability, and vice versa, for the optimal placement of sensors and drivers. These results are applicable to large-scale networks, in which control and monitoring are often sought for small subsets of nodes.

\smallskip
DOI: {\normalfont \href{https://doi.org/10.1109/LCSYS.2023.3289827}{10.1109/LCSYS.2023.3289827}}
\end{abstract}

%=======================================================================
\section{Introduction}

% Intro/motivation
Controllability and observability are  properties that respectively enable full-state control and full-state estimation of a dynamical system. The duality between these properties allow methods developed for feedback controller design to be used for observer design, and vice versa. Beyond classical techniques for pole placement in feedback systems, this duality also finds important applications in optimal control theory \cite{Todorov2008} %(for the design of linear-quadratic regulators and Kalman filters) 
and decentralized control of networked systems \cite{Krtolica1980}. In the context of complex networks, the pressing problem of optimally placing actuators and sensors to ascertain full-state control and monitoring can be solved by a single efficient algorithm \cite{Liu2011} due to the duality between the graph-theoretic notions of \textit{structural} controllability and \textit{structural} observability \cite{Lin1974a}.

Full-state control and estimation are, however, often unfeasible or unneeded in high-dimensional applications such as large-scale networks \cite{Motter2015,Montanari2020}. %Liu2016
Physical, cost, and energy constraints in the placement and operation of actuators and sensors often limit our ability to fully control or observe a network \cite{Pasqualetti2013,Summers2016,Duan2019}. To circumvent these limitations, the generalized notions of output controllability \cite{Lazar2020} and functional observability \cite{Fernando2010} establish the minimal conditions under which part of the state vector (e.g., a \textit{target} subset of state variables) can be controlled and estimated, rather than the full-state vector. These properties enable the control and estimation of target nodes in networks while requiring substantially less resources \cite{Gao2014,Montanari2022}.

% Open problem
The output controllability of a system does not imply in general the functional observability of the dual (transposed) system, which is in contrast with the classical duality between controllability and observability. Consequently, these generalized properties have been studied separately up until now, leading to the independent development of methods for target/output control \cite{Morse1971,Gao2021,Duan2022}, %Parmer2023dynamical
functional observer design \cite{Fernando2010,Rotella2016b}, and actuator/sensor placement \cite{Gao2014,Czeizler2018,Li2020,Li2021,Montanari2022,Li2023}. Yet, a rigorous relation has been recently established between these properties, as characterized by the principles of weak and strong duality \cite{Montanari2023}. In particular, the weak duality establishes that the functional observability of a system implies the output controllability of the dual system, whereas the strong duality establishes that under a particular condition the converse also holds and both properties become equivalent. %is also true. When strong duality holds, both properties are equivalent. 
This opens an opportunity for methods developed for output controllability problems to be mapped to functional observability problems, and vice versa. %\revision{The strong duality principle has already been proven useful as a condition for the design of target controllers in closed-loop systems \cite{Montanari2023}. Notwithstanding, it remains an open question whether this duality can also be leveraged for other applications in system analysis and design, including actuator/sensor placement for the controllability and observability of target variables.}

% Contributions
In this letter, we establish a graph-theoretic characterization of the weak and strong duality principles between target controllability and target observability, which are special notions of output controllability and functional observability for structured systems (Section~\ref{sec.duality}). To this end, we also derive the graph-theoretic conditions for target controllability (Theorem~\ref{thm.targetctrb}), which have been so far restricted to special classes of systems in the literature (Remark~\ref{rem.targetpapers}). As an application of our results, we show that, when strong duality holds, the proposed graph-theoretic characterization enables the use of scalable algorithms to solve both optimal driver and optimal sensor placement in large-scale networks (Section \ref{sec.optimalplac}). The efficacy of our methods in large networks is numerically demonstrated using the \textit{C. elegans} neural network.

%=======================================================================
\section{Preliminaries}
\label{sec.background}

Consider the linear time-invariant dynamical system
\begin{align}
    \dot{\bm x} &= A\bm x + B\bm u, 
    \label{eq.dynsys}
    \\
    \bm y &= C\bm x,
    \label{eq.output}
\end{align}
where $\bm x\in\R^n$ is the state vector, $\bm u\in\R^p$ is the input vector, $\bm y\in\R^q$ is the output vector, $A\in\R^{n\times n}$ is the system matrix, $B\in\R^{n\times p}$ is the input matrix, and $C\in\R^{q\times n}$ is the output matrix. The linear function of the state variables
\begin{equation}
    \bm z = F\bm x
    \label{eq.target}
\end{equation}
defines the \textit{target vector} $\bm z \in \R^r$ sought to be controlled or estimated ($r\leq n$), where $F\in\R^{r\times n}$ is the functional matrix.

The system \eqref{eq.dynsys}--\eqref{eq.target} or, equivalently, the triple $(A,B;F)$ is \textit{output controllable} if, for any initial state $\bm x(0)$ and target state $\bm z^*\in\R^r$, there exists an input $\bm u(t)$ that steers $\bm x(0)$ to some final state $\bm x(t_1)$ satisfying $\bm z(t_1) = F\bm x(t_1) = \bm z^*$ in finite time $t\in [0,t_1]$ \cite{Lazar2020}. A sufficient and necessary condition for this property is given by \cite{Lazar2020}
\begin{equation}
    \rank(F\mathcal C) = \rank(F),
    \label{eq.outputctrb}
\end{equation}
where $\mathcal C = [B \,\, AB \,\, \ldots \,\, A^{n-1} B]$ is the controllability matrix. Despite the terminology ``output'' controllability, note that condition \eqref{eq.outputctrb} is defined for any functional $F$, which is not necessarily related to the output matrix $C$; whether the target variables $\bm z_i(t)$ sought to be controlled are monitored (e.g., measured or estimated) or not depends on the feedback/feedforward control application under consideration.

Moreover, the system \eqref{eq.dynsys}--\eqref{eq.target} or the triple $(C,A;F)$ is \textit{functionally observable} if, for any unknown initial state $\bm x(0)$, there exists a finite time $t_1>0$ such that knowledge of the output $\bm y(t)$ and input $\bm u(t)$ over $t\in [0,t_1]$ suffices to uniquely determine the target state $\bm z(0) = F\bm x(0)$. A sufficient and necessary condition is given by \cite{Jennings2011}
\begin{equation}
    \rank\left(\begin{bmatrix} \mathcal O \\ F \end{bmatrix}\right) = \rank(\mathcal O),
    \label{eq.functobsv}
\end{equation}
where $\mathcal O = [C^\transp \,\, (CA)^\transp \,\, \ldots \,\, (CA^{n-1})^\transp]^\transp$ is the observability matrix. Here, assume that $\rank[C^\transp \,\, F^\transp]^\transp = \rank(C) + \rank(F)$; otherwise, $\bm z_i = \bm\alpha^\transp\bm y$, for some $i$ and $\bm\alpha\in\R^r$, allowing $\bm z_i$ to be trivially estimated without an observer.
%some target variable $\bm z_i$ can be expressed as a linear combination $\bm\alpha^\transp\bm y$ (for some $\bm\alpha\in\R^r$), allowing $\bm z_i$ to be trivially estimated without an observer.}

\begin{comment}
\begin{rem}
    %These properties generalize the classical notions of full-state controllability and observability. 
    If all variables are sought to be controlled or estimated (i.e., $F=I_n$, where $I_n$ is the identity matrix of size $n$), then conditions \eqref{eq.outputctrb} and \eqref{eq.functobsv} respectively reduce to the well-known tests for full-state controllability ($\rank(\mathcal C) = n$) and full-state observability ($\rank(\mathcal O) = n$).
\end{rem}
\end{comment}

In spite of the duality between the (full-state) observability of a system $(C,A)$ and the controllability of the dual system $(C^\transp, A^\transp)$, functional observability and output controllability are not dual properties in general when $\rank(F)< n$~\cite{Montanari2023}. To see this, consider a pair of dynamical systems $(C,A;F)$ and $(A^\transp,C^\transp;F)$, where $\mathcal O$ is the observability matrix of the former system and $\mathcal C = \mathcal O^\transp$ is the controllability matrix of the latter.
Note that 
%condition \eqref{eq.functobsv} is equivalent to $\operatorname{row}(F)\subseteq\operatorname{row}(\mathcal O)$ and that 
condition \eqref{eq.outputctrb} is equivalent to $\rank(F\mathcal O^\transp) = \rank (F)$ for a triple $(A^\transp,C^\transp;F)$. Thus, it follows that any system $(C,A;F)$ that satisfies condition \eqref{eq.functobsv} also satisfies condition \eqref{eq.outputctrb} for the dual $(A^\transp,C^\transp;F)$. The converse, however, is not always true.
%Note that condition \eqref{eq.outputctrb} for the system $(A^\transp,C^\transp;F)$, being equivalent to $\rank(F\mathcal O^\transp) = \rank (F)$, is weaker (easier to satisfy) than condition \eqref{eq.functobsv} for $(C,A;F)$. 
As a consequence, $(A^\transp,C^\transp;F)$ may be output controllable without necessarily implying that $(C,A;F)$ is functionally observable (see Example~\ref{examp.graph} below). 
\section{Target Controllability and Observability}
\label{sec.duality}

We show that the relation and equivalence between output controllability and functional observability are characterized by the notions of weak and strong duality. 
This duality follows directly from an intuitive graph-theoretic representation of output controllability and functional observability, which allows us to explicitly leverage the structure of the system matrix $A$ and its inputs, outputs, and target variables (defined by matrices $B$, $C$, and $F$, respectively). Before stating our results, we first define graph concepts for structured systems.

% The relation and equivalence between output controlla-
% bility and functional observability are characterized by the
% notions of weak and strong duality. In this section, we present
% an intuitive graph-theoretic interpretation of these principles
% that explicitly leverages the structure of the system matrix A
% and the sets of inputs, outputs, and target variables (defined
% by matrices B, C, and F , respectively). Before presenting
% this duality, we first define some graph concepts for struc-
% tured systems and establish the graph-theoretic conditions for
% output controllability and functional observability.

%-----------------------------------------------------------------------
\subsection{Structured systems and graph theory}

    A matrix $M\in\{0,\star\}^{m\times n}$ is a \textit{structured matrix} if $M_{ij}$ is either a fixed zero entry or an independent nonzero entry, denoted by a $\star$. A matrix $\tilde M$ is a numerical realization of $M$ if real numbers are assigned to all nonzero entries of $M$.

The \textit{inference graph} of a system \eqref{eq.dynsys}--\eqref{eq.target} is denoted by $\mathcal G(A,B,C;F)=\{\mathcal V,\mathcal E\}$, where $\mathcal V = \mathcal X\cup\mathcal U\cup \mathcal Y$ is the set of nodes, $\mathcal E = \mathcal E_{\mathcal X}\cup \mathcal E_{\mathcal U}\cup \mathcal E_{\mathcal Y}$ is the set of edges, and $(A,B,C,F)$ are structured matrices. 
Nodes represent state variables $\mathcal X = \{\bm x_1,\ldots,\bm x_n\}$, inputs $\mathcal U = \{\bm u_1,\ldots,\bm u_p\}$ (driver nodes), and outputs $\mathcal Y = \{\bm y_1,\ldots,\bm y_q\}$ (sensor nodes). Let $(\bm x_i,\bm x_j)\in\mathcal E_{\mathcal X}$ (directed edge from $\bm x_j$ to $\bm x_i$) if $A_{ij}\neq 0$, $(\bm x_i,\bm u_j)\in\mathcal E_{\mathcal U}$ if $B_{ij}\neq 0$, and $(\bm y_i,\bm x_j)\in\mathcal E_{\mathcal Y}$ if $C_{ij}\neq 0$. The set of \textit{target nodes} $\mathcal T\subseteq\mathcal X$ defines a set of state variables sought to be controlled or estimated, where $\bm x_j\in\mathcal T$ if $F_{ij}\neq 0$ for some $i$. The inference graph is denoted simply by $\mathcal G(A,B;F)$ and $\mathcal G(C,A;F)$ when considering the output controllability and functional observability of a triple, respectively.

A subset of nodes $\mathcal V'\subseteq\mathcal X$ has a \textit{dilation} in a graph $\mathcal G$ if $|P(\mathcal V')|<|\mathcal V'|$, where $|\cdot|$ denotes the set cardinality and $P(\mathcal V')$ is the set of all nodes $v_i\in\mathcal X\cup\mathcal U$ that have a direct link to $\mathcal V'$ (i.e., the set of predecessors of $\mathcal V'$).
Similarly, $\mathcal V'\subseteq\mathcal X$ has a \textit{contraction} in $\mathcal G$ if $|S(\mathcal V')|<|\mathcal V'|$, where $S(\mathcal V')$ is the set of all nodes $v_i\in\mathcal X\cup\mathcal Y$ that have a direct link from $\mathcal V'$ to $v_i$ (i.e., the set of successors of $\mathcal V'$). 
Let $\mathcal D_k$ ($\mathcal K_k$) be a \textit{minimal dilation (contraction) set} of $\mathcal G$ if $\mathcal D_k$ ($\mathcal K_k$) has a dilation (contraction) and no subset $\mathcal D'_k\subset \mathcal D_k$ ($\mathcal K_k'\subset \mathcal K_k$) has a dilation (contraction).
%, for all subsets $\mathcal D'_k\subset \mathcal D_k$ ($\mathcal K_k'\subset \mathcal K_k$), the subset $\mathcal D'_k$  ($\mathcal K'_k$) has no dilations (contractions).

\begin{rem}
    $\mathcal G(A,B)$ has a dilation if there is a set of $k$ rows of $[A\,\, B]$ that contains nonzero entries in less than $k$ columns of the submatrix formed by these $k$ rows. In fact, $\mathcal G(A,B)$ has a dilation if and only if $\rank[A \,\, B] < n$ \cite{Brisk1984}.
\label{rem.dilation}
\end{rem}

We now revisit a fundamental result on the controllability \cite{Lin1974a} and, by duality, observability of structured systems.

\begin{defin} \label{def.structuralctrb}
    The structured system $(A,B)$ [$(C,A)$] is~structurally controllable [observable] if there exists a numerical realization $(\tilde A,\tilde B)$ [$(\tilde C,\tilde A)$] that is controllable [observable].
\end{defin}

\begin{thm} {\normalfont \cite{Lin1974a}} \label{thm.structctrb}
    The system $(A,B)$ [or $(C,A)$] is structurally controllable [or observable] if and only if $\mathcal G(A,B)$ [or $\mathcal G(C,A)$] satisfies the following conditions:
    \begin{enumerate}
        \item for each state variable $\bm x_i\in\mathcal X$, there exists a path from some driver node $\bm u_i\in\mathcal U$ to $\bm x_i$ [or every $\bm x_i\in\mathcal X$ has a path to some sensor node $\bm y_i\in\mathcal Y$];
        \item $\mathcal G$ has no dilations [or contractions].
    \end{enumerate}
\end{thm}

%-----------------------------------------------------------------------
\subsection{Target controllability and target observability}

We now establish the graph-theoretic conditions for output controllability and functional observability. These conditions are presented for systems in which nodes are independently driven, measured, and targeted, as formalized below.

\begin{assump} \label{assump.independent}%
    We assume that each column of $B$ and each row of $C$ and $F$ have a single nonzero entry. We also assume that one of the following graph-theoretical conditions on the structured matrix $A$ is satisfied: (i) there exists some numerical realization $\tilde A$ that is diagonalizable; (ii) $A_{ii} \neq 0$ for every target node $\bm x_i\in\mathcal T$. We note that the assumption on $A$ can be relaxed to a weaker algebraic condition based on the Jordan form, which will be presented in future work.
\end{assump}

\begin{defin} \label{def.targetctrb}
    The structured system $(A,B;F)$ is target controllable if there exists some numerical realization $(\tilde A,\tilde B;\tilde F)$ that is output controllable. Likewise, the structured system $(C,A;F)$ is target observable if there exists some numerical realization $(\tilde C,\tilde A;\tilde F)$ that is functionally observable.
\end{defin}

Given the large adoption of the term ``target controllability'' by the community \cite{Gao2014,Waarde2017,Czeizler2018,Li2020,Moothedath2019,Li2021,Li2023} and the duality between output controllability and functional observability \cite{Montanari2023}, it seems appropriate to unify these two structural properties under a common nomenclature\textemdash \textit{target controllability} and \textit{target observability}\textemdash as in Definition~\ref{def.targetctrb}.

%---------------------------------
We present the following theorem on target controllability, which establishes graph-theoretic conditions equivalent to condition \eqref{eq.outputctrb} for a structured system $(A,B;F)$.

\begin{thm} %{\normalfont\bf (Target controllability)}
\label{thm.targetctrb}
    The system $(A,B;F)$ is target controllable if and only if $\mathcal G(A,B;F)$ satisfies the following conditions:
    \begin{enumerate}
        \item for each target node $\bm x_i\in\mathcal T$, there exists a path from some driver node $\bm u_i\in\mathcal U$ to $\bm x_i$;
        \item no subset $\mathcal T_\ell\subseteq\mathcal T$ in $\mathcal G'(A,B;F)$ has a dilation, where $\mathcal G'$ is a subgraph of $\mathcal G$ containing all possible paths from every $\bm u_i\in\mathcal U$ to any $\bm x_i\in\mathcal T$.
    \end{enumerate}
\end{thm}

\begin{proof} 
See Appendix.
\end{proof}

\begin{rem}     \label{rem.targetpapers}
    Theorem~\ref{thm.targetctrb} generalizes previous results on target controllability \cite{Blackhall2010,Gao2014,Moothedath2019}, as shown next. Assume $\bm u_1$ has a path to all $\bm x_i\in\mathcal T$. For directed tree graphs $\mathcal G(A)$, a system is target controllable if and only if the path length from a driver node to each target node is unique \cite[Th.~2]{Gao2014}, which is equivalent to $\mathcal T_\ell\subseteq\mathcal T$ having no dilations in $\mathcal G'$ since $\mathcal G'$ is also a directed tree and thus has no cycles. For systems with single-input matrix $B$, target controllability holds if $\mathcal G'$ has a perfect matching \cite[Th.~2]{Moothedath2019}, which is sufficient for $\mathcal G'$ to have no dilations, satisfying condition 2 of Theorem~\ref{thm.targetctrb}. Likewise, a multiple-input system is target controllable if $\mathcal G$ can be covered by a union of cacti structures \cite[Th. 17]{Blackhall2010}, which is sufficient for $\mathcal G'$ to have no dilations. Note that Refs. \cite{Blackhall2010,Moothedath2019} established only sufficient conditions.

    %It is clear this condition is only sufficient. Consider a graph 1 -> 2 -> 3, and 2 -> 4 -> 5. Let {3,5} be targets. Clearly, the network is target controllable (it satisfies Theorem 2, Ref. 12) but has not spanning cacti covering all target nodes.
    
    The related studies~\cite{Waarde2017,Li2021} on target controllability are complementary to our results, providing conditions for less generic types of structured systems (e.g., symmetric matrices \cite{Li2021}) or for a stronger notion of target controllability in~which \textit{all} (rather than \textit{some}, as in Definition~\ref{def.targetctrb}) numerical realizations $(\tilde A,\tilde B; \tilde F)$ are output controllable \cite{Waarde2017}.
\end{rem}

%---------------------------------
The graph-theoretic conditions for target observability have already been established in Ref.~\cite{Montanari2022} (under the nomenclature of ``structural functional observability''), being equivalent to condition \eqref{eq.functobsv} for a structured system $(C,A;F)$.

\begin{thm} {\normalfont \cite{Montanari2022}} \label{thm.targetobsv}
    The system $(C,A;F)$ is target observable if and only if $\mathcal G(C,A;F)$ satisfies the following conditions:
    \begin{enumerate}
        \item every target node $\bm x_i\in\mathcal T$ has a path to some sensor node $\bm y_i\in\mathcal Y$;
        \item $\mathcal T\cap\mathcal K = \emptyset$, where  $\mathcal K = \bigcup_k \mathcal K_k$ is the union of all minimal contraction sets in $\mathcal G(C,A;F)$.
    \end{enumerate}
\end{thm}

%-----------------------------------------------------------------------
\subsection{Duality principle}

We now establish the weak and strong duality principles for target controllability and target observability.
To this end, consider a pair of structured systems $(C,A;F)$ and $(A^\transp,C^\transp;F)$ and their inference graphs $\mathcal G(C,A;F) = \{\mathcal X\cup\mathcal Y,\mathcal E_{\mathcal X}\cup\mathcal E_{\mathcal Y}\}$ and $\mathcal G(A^\transp,C^\transp;F) =  \{\mathcal X\cup\mathcal U,\mathcal E_{\mathcal X}\cup\mathcal E_{\mathcal U}\}$.

\begin{rem}
    $\mathcal G(A^\transp,C^\transp;F)$ is equivalent to graph $\mathcal G(C,A;F)$ with reversed edges and $\mathcal U = \mathcal Y$. Moreover, a set $\mathcal V'\subseteq \mathcal X$ has a dilation in $\mathcal G(A^\transp,C^\transp;F)$ if and only if $\mathcal V'$ has a contraction in $\mathcal G(C,A;F)$. This is later illustrated in Fig.~\ref{fig.graph}.
\label{rem.dualgraph}
\end{rem}

\begin{thm} {\normalfont \textbf{(Weak duality)}} \label{thm.weakduality}
    If $(C,A;F)$ is target observable, then $(A^\transp, C^\transp; F)$ is target controllable.
\end{thm}

\begin{proof}
    Since $(C,A;F)$ is target observable, the conditions of Theorem~\ref{thm.targetobsv} are satisfied. 
    First, given Remark~\ref{rem.dualgraph}, if condition 1 of Theorem~\ref{thm.targetobsv} holds, then for each $\bm x_i\in\mathcal T$ in the reversed graph $\mathcal G(A^\transp,C^\transp;F)$ there exists a path from 
    some $\bm u_i\in\mathcal U$ to $\bm x_i$, satisfying condition 1 of Theorem~\ref{thm.targetctrb}. 
    Second, it follows from Remark~\ref{rem.dualgraph} that $\mathcal K_k = \mathcal D_k$, $\forall k$, where $\mathcal K_k$ and $\mathcal D_k$ are minimal contraction and dilation sets in $\mathcal G(C,A;F)$ and $\mathcal G(A^\transp,C^\transp;F)$, respectively. By induction, $\mathcal K = \mathcal D = \bigcup_k \mathcal D_k$.
    Since $\mathcal T\cap\mathcal K =\emptyset$ holds in $\mathcal G(C,A;F)$, it follows that $\mathcal T\cap \mathcal D = \emptyset$ also holds in $\mathcal G(A^\transp,C^\transp;F)$. 
    If $\mathcal T \cap \mathcal D = \emptyset$, then $\mathcal D_k\not\subseteq\mathcal T_\ell$ for any subset $\mathcal T_\ell\subseteq\mathcal T$. Thus, no subset $\mathcal T_\ell\subseteq\mathcal T$ has a dilation in $\mathcal G$ and hence in $\mathcal G'(A^\transp,C^\transp;F)$. This satisfies condition 2 of Theorem~\ref{thm.targetctrb} and so $(A^\transp, B^\transp;F)$ is target controllable. 
\end{proof}

\begin{thm} {\normalfont \textbf{(Strong duality)}} \label{thm.strongduality}
    The system $(C,A;F)$ is target observable if and only if $(A^\transp,C^\transp;F)$ is target controllable and $\mathcal T\cap\mathcal D = \emptyset$, where $\mathcal D = \bigcup_k \mathcal D_k$ is the union of all minimal dilation sets $\mathcal D_k$ in $\mathcal G(A^\transp,C^\transp;F)$.
\end{thm}

\begin{proof}
    We show that Theorems~\ref{thm.targetctrb} and \ref{thm.targetobsv} are equivalent for graphs $\mathcal G(C,A;F)$ and $\mathcal G(A^\transp,C^\transp;F)$ under the stated conditions. The equivalence between conditions 1 of Theorems~\ref{thm.targetctrb} and \ref{thm.targetobsv} follows directly from Remark~\ref{rem.dualgraph}. 
    It also follows from Remark~\ref{rem.dualgraph} that $\mathcal D_k = \mathcal K_k$, $\forall k$, and $\mathcal D = \mathcal K$. Thus, condition 2 of Theorem~\ref{thm.targetobsv} is equivalent to $\mathcal T\cap\mathcal D = \emptyset$.
\end{proof}

\begin{rem} \label{rem.allnodestargets}
    When all state variables are targeted ($\mathcal T = \mathcal X$), Theorems~\ref{thm.targetctrb}~and~\ref{thm.targetobsv} reduce to Theorem~\ref{thm.structctrb}. This is evident for condition 1 of both theorems. For condition 2 of Theorem~\ref{thm.targetctrb}, when $\mathcal T = \mathcal X$, $\mathcal G' = \mathcal G$ and thus $\mathcal G$ must have no dilations. For condition 2 of Theorem~\ref{thm.targetobsv}, when $\mathcal T = \mathcal X$ it follows that $\mathcal T\cap \mathcal K = \emptyset$ if and only if $\mathcal K = \emptyset$, implying that $\mathcal G$ must have no contractions. Thus, the strong duality reduces to the classical duality between (structural) controllability and  observability.
\end{rem}

\begin{rem}
    For many sparse directed networks, the strong duality condition $\mathcal T \cap \mathcal D = \emptyset$ can be computationally tested in $\mathcal G(A^\transp,C^\transp;F)$ as follows. For every target node $\bm x_i\in\mathcal T$, a breadth-first search algorithm can be used to build the set of nodes $\mathcal S\subseteq{\mathcal X\cup\mathcal U}$ composed of the union of sets $S(P(\bm x_i))$, $S(P(S(P(\bm x_i))))$, and so forth, incurring in a computational complexity of order $O(n+|\mathcal E|)$.
    The existence of a minimal dilation set $\mathcal D_k\supseteq \{\bm x_i\}$ can then be readily verified by testing the condition $|P(\mathcal S')|<|\mathcal S'|$ for all possible subsets $\mathcal S'\subseteq\mathcal S$.
    This procedure is feasible if $\mathcal S$ is sufficiently small, which holds in general for high-dimensional networks when $\mathcal G$ has few cycles, small node degrees ($\sum_j {A_{ij}} \ll n$, $\forall i$), and few targets ($r\ll n$).
    For undirected networks, however, $\mathcal S = \mathcal X$, making this test computationally expensive for large $n$.
\label{rem.computation}
\end{rem}

A sufficient condition based on the structure of the inference graph $\mathcal G(A)$ is provided below for strong duality.

\begin{cor} \label{cor.selfloop}
    $(C,A;F)$ is target observable if $(A^\transp, C^\transp; F)$ is target controllable and every $\bm x_i\in\mathcal T$ has a self-edge.
\end{cor}

\begin{proof}
    If a target node has a self-edge, then it does not belong to a minimal contraction/dilation set. Since this holds for all $\bm x_i\in\mathcal T$,  conditions 2 of Theorems~\ref{thm.targetctrb} and \ref{thm.targetobsv} are satisfied for graphs $\mathcal G(A^\transp,C^\transp;F)$ and $\mathcal G(C,A;F)$, respectively.
\end{proof} \vspace{-0.3cm}

\begin{figure}     
\centering
    \includegraphics[width=0.7\columnwidth]{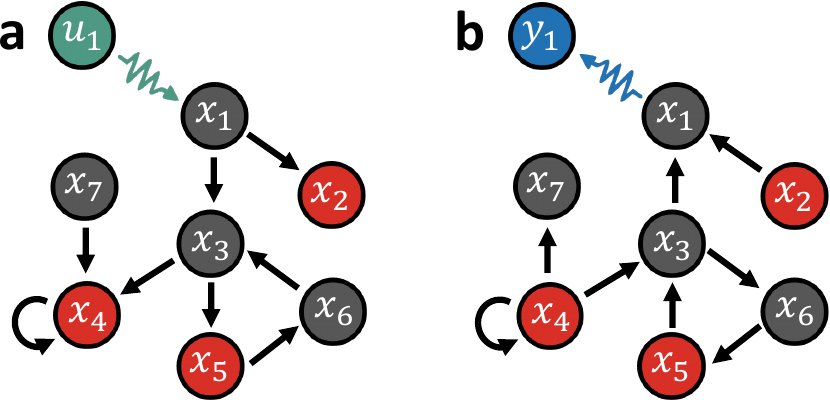}
    \caption{\label{fig.graph} Inference graphs of a dual pair of dynamical systems. (a) Target controllability of a structured system $(A,B;F)$. (b) Target observability of the dual structured system $(B^\transp,A^\transp;F)$. The driver, sensor, and target nodes are indicated in green, blue, and red, respectively. 
    % \revision{(c) Example of a (random) numerical realization $(\tilde A,\tilde B;\tilde F)$ of the structured system $(A,B;F)$, and its corresponding tests for output controllability and functional observability.}
    } \vspace{-0.5cm}
\end{figure}

%-----------------------------------------------------------------------
\begin{example} \label{examp.graph}
Consider the dual pair of systems illustrated in Fig.~\ref{fig.graph}. System $(A,B;F)$ is target controllable since there is a path from $u_1$ to every target node and no subset $\mathcal T_\ell\subseteq\mathcal T$ has a dilation in $\mathcal G' = \mathcal G\backslash\{\bm x_7\}$ (e.g., for $\mathcal T_\ell = \{\bm x_4,\bm x_5\}$, we have that $P(\mathcal T_\ell) = \{\bm x_3,\bm x_4\}$ and thus $|P(\mathcal T_\ell)|=|\mathcal T_\ell|$). However, the dual system  $(B^\transp,A^\transp;F)$ is not target observable since $\mathcal D_k = \{\bm x_2,\bm x_3\}$ is a minimal dilation set and $\mathcal T\cap\mathcal D = \{\bm x_2\}$, hence strong duality does not hold. There are several ways to enforce strong duality and make $(B^\transp,A^\transp;F)$ target observable; for example, by adding a self-edge to $\bm x_2$, connecting a second driver node $\bm u_2$ to $\bm x_2$, or removing $\bm x_2$ from the set of target nodes (corresponding to changes in the structure of matrices $A$, $B$, and $F$, respectively). 
%\revision{Following Definition~\ref{def.targetctrb}, Fig.~\ref{fig.graph}c illustrates that the conditions for target controllability and target observability are generic and hold for all numerical matrices $(\tilde A,\tilde B;\tilde F)$ sharing the structure of $(A,B;F)$ except for a set of matrices with Lesbegue measure zero.}
Following Definition~\ref{def.targetctrb}, the conditions for target controllability and target observability are generic and hold
for all numerical matrices $(\tilde A,\tilde B;\tilde F)$ sharing the structure of $(A,B;F)$ except for a set of matrices of Lesbegue measure zero.
\end{example}

%=======================================================================
\section{Optimal Driver and Sensor Placement}
\label{sec.optimalplac}

%-----------------------------------------------------------------------
\subsection{Duality and algorithms}

The weak duality principle shows that methods developed for target observability problems can be directly applied to target controllability problems (by using the dual graph),~as well as the converse when strong duality holds. Such methods include algorithms designed to test the conditions of Theorems \ref{thm.targetctrb} and \ref{thm.targetobsv} for high-dimensional systems, or\textemdash as we consider next\textemdash to find a minimum set of driver nodes $\mathcal U$ for target controllability or sensor nodes $\mathcal Y$ for target observability. The latter are respectively addressed as the problems of minimum driver placement for target controllability (MDPt) and minimum sensor placement for target observability (MSPt).

Thus far, no algorithm has been developed to solve the MSPt for \textit{general} inference graphs $\mathcal G(A)$ due to the computational challenges in verifying condition 2 of Theorem~\ref{thm.targetobsv} for generic (possibly undirected) graphs (Remark~\ref{rem.computation}). For a broad class of applications where every target node has a self-edge in $\mathcal G(A)$ (Corollary~\ref{cor.selfloop}), the MSPt can be formulated as a set cover problem, which can be approximately solved by combining a greedy algorithm and breadth-first searches, as presented in \cite[Alg. 1]{Montanari2022}. Owing to the weak duality principle, it follows that the MDPt can also be formulated as a set cover problem for the dual graph $\mathcal G(A^\transp)$ and solved by the same algorithm when every target node has a self-edge.

\begin{figure}[t] 
    \centering
    \includegraphics[width=0.88\columnwidth]{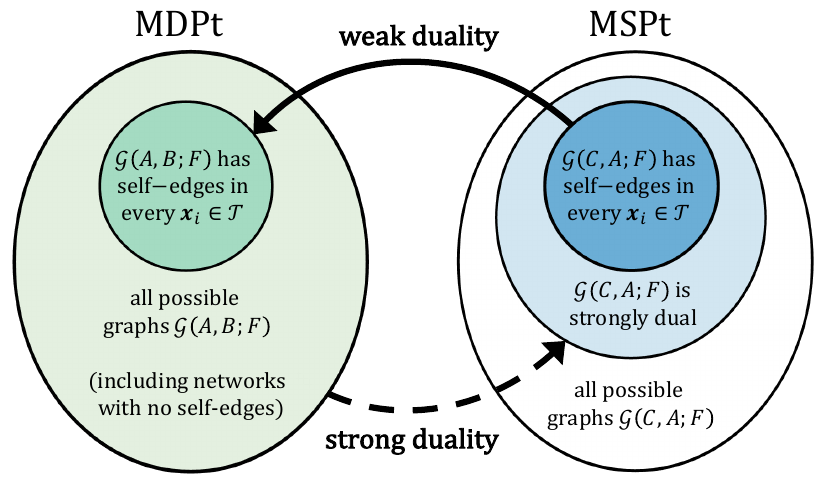}
    \caption{\label{fig.algorithmset} Sets of MDPt and MSPt problems. Depending on the structure of the inference graph $\mathcal G(A)$, algorithms developed to solve a MDPt problem can be employed to solve a dual MSPt problem, and vice versa. The origin of the dashed and solid arrows are the set of  problems that were originally solved by the algorithms presented in Refs.~\cite{Gao2014} and \cite{Montanari2022}, respectively. The endpoints of these arrows indicate the new sets of problems that these algorithms can solve due to the weak and strong duality principles. The light (dark) shades represent weakly (strongly) dual sets of problems.} \vspace{-0.5cm}
\end{figure}

When strong duality holds, we show that a new class of problems can be solved. Unlike the MSPt, the MDPt can be solved efficiently (though approximately) for \textit{any} inference graph $\mathcal G(A)$. This is enabled by the fact that condition 2 of Theorem~\ref{thm.targetctrb} is weaker than condition 2 of Theorem~\ref{thm.targetobsv}, which allows it to be enforced using a greedy algorithm that recursively solves a maximum matching problem in an induced bipartite graph, as proposed in \cite[Alg. 3]{Gao2014}. The strong duality principle thus enables this MDPt algorithm to be employed for MSPt problems, providing an efficient (approximate) solution for the set of all $\mathcal G(C,A;F)$ that satisfy the strong duality condition in Theorem~\ref{thm.strongduality}. 
% Since every strongly dual $\mathcal G(C,A;F)$

Fig.~\ref{fig.algorithmset} summarizes the relation between the MDPt and MSPt problems, illustrating how the weak and strong duality principles can be applied for the conversion of algorithms from one problem to the other. The light green and dark blue sets are those containing the problems originally solved in Refs.~\cite{Gao2014} and \cite{Montanari2022}, respectively. It is now evident that the strong duality principle enables the translation of algorithms to solve a new class of problems (contained in the light blue region) that did not have a solution available in the literature yet. The white region remains as the most general set of MSPt problems with no available solvers.

%-----------------------------------------------------------------------
\subsection{Numerical results}

Fig.~\ref{fig.complexnet} illustrates the MDPt and MSPt problems applied to a high-dimensional system, the \textit{C. elegans} neural network. %For demonstration purposes,
The network is modeled as a linear system \eqref{eq.dynsys} where each variable $\bm x_i$ represents a neuron (node) and $A$ is the adjacency matrix. Given the highly directed and sparse nature of the inference graph $\mathcal G(A)$ and the small number of selected target nodes ($r=0.05n$), the presence of minimal dilation sets containing $\mathcal T$ can be efficiently tested following Remark~\ref{rem.computation}. For the set of target nodes $\mathcal T$ shown in Fig.~\ref{fig.complexnet}a, it holds that $(C,A;F)$ satisfies $\mathcal T\cap\mathcal K = \emptyset$ for any choice of $C$ and, therefore, the system is strongly dual. The network has no self-edges, implying that this MSPt problem falls into the class of problems that can be solved by the MDPt algorithm (light blue set in Fig.~\ref{fig.algorithmset}).
Fig.~\ref{fig.complexnet}a shows the minimum set of drivers and sensors selected with \cite[Alg. 3]{Gao2014} by considering the original graph $\mathcal G(A)$ and the dual graph $\mathcal G(A^\transp)$, respectively.
The algorithm provides an efficient approximation, in which the minimum number of sensors and drivers correspond to only 1\% and 1.5\% of the network size, respectively.
As the number of targets increases, Fig.~\ref{fig.complexnet}b shows that the number of drivers and sensors remain relatively small compared to the network size, as also observed in other complex networks without and with self-edges (cf. \cite[Fig. 6]{Gao2014} and \cite[Fig. 2]{Montanari2022}).

Algorithmic implementations to solve the MDPt \cite[Alg. 3]{Gao2014} and MSPt \cite[Alg. 1]{Montanari2022} problems for arbitrary inference graphs $\mathcal G(A)$ and target sets $\mathcal T$ are available at \href{https://github.com/montanariarthur/TargetCtrb}{https://github.com/montanariarthur/TargetCtrb}. Beyond the placement of drivers and sensors, our GitHub repository also provides code on how to effectively design feedback controllers \cite{Montanari2023} and functional observers \cite{Fernando2010,Montanari2022} %Trinh2012
for the stable control and estimation of target variables, respectively.

\begin{figure} 
    \centering
    \includegraphics[width=0.88\columnwidth]{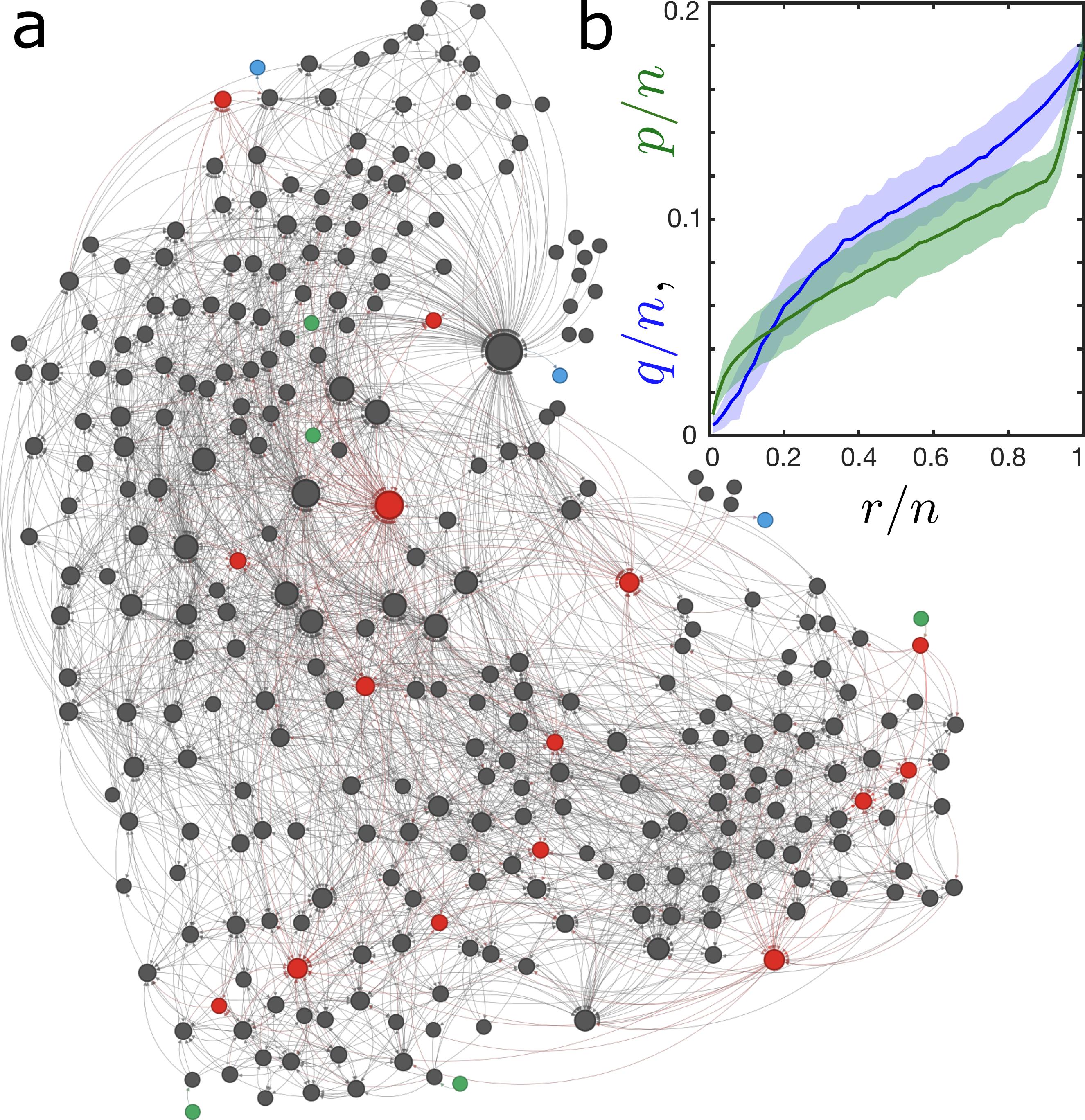}
    \caption{\label{fig.complexnet} (a) Optimal driver and sensor placement for the target controllability and target observability of the \textit{C. elegans} neural network. The driver, sensor, and target nodes are indicated in green, blue, and red, respectively. %\revision{The node size is proportional to the number of in-degree links.}
    (b) Minimum number of drivers $p$ (green) and sensors $q$ (blue) as a function of the number of target nodes $r$ (normalized by the network size $n$). Each data point is an average over 100 realizations of randomly selected target nodes, where shaded areas indicate three standard deviations.
    } \vspace{-0.5cm}
\end{figure}

%=======================================================================
\section{Conclusion}

Examining the rank-based conditions for output controllability and functional observability, it is not immediately clear for which systems the output controllability of $(A,B;F)$ implies the functional observability of $(B^\transp,A^\transp;F)$. For network applications where target variables are independently sought to be controlled or estimated (Assumption~\ref{assump.independent}), our results provide a graph-theoretic characterization of target controllability and target observability.
Unlike output controllability and functional observability, each characterized by a single rank-based condition, target controllability and target observability are individually depicted by two graph-based conditions %: (1) the existence of paths between target nodes and driver (sensor) nodes and (2) the presence of dilations (contractions) involving target nodes. 
that highlight the weak and strong dualities between these properties. The first condition\textemdash related to the existence of paths from/to target nodes to/from sensor/driver nodes\textemdash is equivalent for any dual pair of inference graphs. However, the second condition\textemdash related to dilations and contractions in a graph\textemdash is inherently stronger for target observability than for target controllability. In particular, it follows from Theorems~\ref{thm.structctrb}--\ref{thm.strongduality} that the set of structurally observable ((and, equivalently, the dual set of structurally controllable) systems are contained inside the set of target observable systems, which in turn are contained inside the dual set of target controllable systems.

Our application of an MDPt algorithm for a class of MSPt problems is one of many possible uses of the established duality principle. Here, we focused on algorithms proposed in Refs.~\cite{Gao2014,Montanari2022} due to their intrisic connection to the graph-theoretic conditions in Theorems \ref{thm.targetctrb} and \ref{thm.targetobsv}. Nonetheless, we expect that many other methods developed for the broadly explored problem of target controllability (based on graph theory \cite{Li2020,Li2023}, linear programming \cite{Waarde2017}, or structural rank conditions \cite{Czeizler2018,Moothedath2019}) may also find new applications in functional observability problems, as well as the converse. %We also expect the duality proposed in this work to not be restricted to linear structured systems and potentially be generalized to broader classes of dynamical systems in which the notions of output controllability and functional observability have already been established  \cite{Tie2020,Montanari2022a}.

%=======================================================================
\appendix

\section{Proof of Theorem~\ref{thm.structctrb}}
\label{app.proof}

\noindent
\textit{Proof of Theorem~\ref{thm.targetctrb}.} Let $\mathcal X_1\subseteq\mathcal X$ be the set of all state variables belonging to a path in $\mathcal G$ from some driver node $\bm u_i\in\mathcal U$ to some target node $\bm x_i\in\mathcal T$, and $\mathcal X_2 = \mathcal X\backslash \mathcal X_1$ be the complement set. Define $|\mathcal X_1| = k$ and $|\mathcal X_2| = n-k$.

\textit{Sufficiency.} Suppose that condition 1 is satisfied, i.e., $\mathcal T\subseteq\mathcal X_1$ and $k\geq r$. After applying a permutation of coordinates such that the nodes in $\mathcal X_1$ appear first, we have the form
\begin{equation}
    A = \begin{bmatrix} A_{11} & A_{12} \\ A_{21} & A_{22} \end{bmatrix}, \,\,
    B = \begin{bmatrix} B_{1} & 0 \\ 0 & B_2 \end{bmatrix}, \,\, 
    F = \begin{bmatrix} F_1 & 0 \end{bmatrix},
\end{equation}
where $A_{11}\in\R^{k\times k}$, $A_{22}\in\R^{(n-k)\times (n-k)}$, $B_1\in\R^{k\times p_1}$, $B_2\in\R^{(n-k)\times (p-p_1)}$, and $F_1\in\R^{r\times k}$. Matrices $A_{12}$ and $A_{22}$ correspond to paths between sets $\mathcal X_1$ and $\mathcal X_2$, $B_1$ corresponds to all $p_1$ driver nodes $\bm u_i \in \mathcal U_1\subseteq\mathcal U$ that have some path to a target node, and $B_2$ to $\bm u_i\in\mathcal U_2 = \mathcal U\backslash\mathcal U_1$. This yields the subgraph $\mathcal G' = \{\mathcal X_1,\mathcal E_1\}$, where $\mathcal E_1= \mathcal E_{\mathcal X_1}\cup\mathcal E_{\mathcal U_1}$, $(\bm x_i,\bm x_j)\in\mathcal E_{\mathcal X_1}$ if $[A_{11}]_{ij}\neq 0$ , and $(\bm x_i,\bm u_j)\in\mathcal E_{\mathcal U_1}$ if $[B_{1}]_{ij}\neq 0$.

The controllability matrix $\mathcal C$ of system $(A,B)$ has the form
\begin{equation}
    \mathcal C =
    \begin{bmatrix}
        B_1 & 0 & A_{11}B_1 & 0 & A_{11}^2 B_1 & 0 & \ldots\\
        0 & B_2 & A_{21}B_1 & A_{22}B_2 & XB_1 & A_{22}^2B_2 & \ldots
    \end{bmatrix},
\label{eq.proof.ctrbsubmatrix}
\end{equation}
where $X = A_{21}A_{11}+A_{22}A_{21}$. In Eq.~\eqref{eq.proof.ctrbsubmatrix}, we have used the fact that $A_{12}A_{22}^k B_2 = 0$, $\forall k\in\{0,1,\ldots\}$, since by definition no driver node $\bm u_i\in\mathcal U_2$ has a path to some node $\bm x_i\in\mathcal X_1$. Likewise, matrices of form $A_{12}A_{21} B_1$ and $A_{12}A_{21}A_{11}B_1$ are zero; otherwise, there would exist a path from $\bm u_i\in\mathcal U_1$ to $\bm x_i\in\mathcal X_1$ passing by a node $\bm x_j\in\mathcal X_2$, which contradicts the assumption that all such paths are already covered in $\mathcal G' = \{\mathcal X_1,\mathcal E_1\}$.
Therefore, it follows by construction that $F\mathcal C = F_1\mathcal C_1$, where $\mathcal C_1$ is the controllability matrix of pair $(A_{11},B_1)$. It remains to show that if no subset $\mathcal T_\ell\subseteq\mathcal T$ has a dilation in $\mathcal G'$ then $\rank(F_1\mathcal C_1) = r$.

Assume without loss of generality that the first $r$ nodes in $\mathcal X_1$ belong to $\mathcal T$. First, suppose no driver node is directly connected to a target node in $\mathcal G'$ (first $r$ rows of $B_1$ are zero). Following Remark~\ref{rem.dilation}, if no subset $\mathcal T_\ell\subseteq\mathcal T$ has a dilation in $\mathcal G'$, it follows that the first $r$ rows of $A_{11}$ have nonzero entries in at least $r$ columns of the submatrix formed by these $r$ rows. Since $A$ satisfies Assumption \ref{assump.independent} and there exists a path from some driver node in $\mathcal U_1$ to every target node in $\mathcal T$, it follows that the first $r$ rows of $\mathcal C_1$ also have nonzero entries in at least $r$ columns. Therefore, $\rank(F_1\mathcal C_1) = r$.
Second, suppose a driver node is directly connected to target node $\bm x_1$. Given that $B_1$ has a single nonzero entry per column (Assumption~\ref{assump.independent}), $\bm x_1$ does not belong to a minimal dilation set and the first row of matrix $F_1\mathcal C_1$ is always linearly independent from the other rows. Therefore, $\rank(F_1\mathcal C_1) = \rank(F_1'\mathcal C_1') + 1$, where $F_1'\mathcal C_1'$ is a submatrix of $F_1\mathcal C_1$ without the first row. The rest of the proof follows as above for the submatrix $F_1'\mathcal C_1'$.

\textit{Necessity.} The necessity of condition 2 follows from the fact that that if some subset $\mathcal T_\ell\subseteq \mathcal T$ has a dilation, then the first $r$ rows of $A_{11}$ have nonzero entries in less than $r$ columns, and so does $\mathcal C_1$. This implies that $\rank(F_1\mathcal C_1) < r$.

For the necessity of condition 1, suppose there are no paths from driver nodes $\mathcal U$ to some nodes $\mathcal X_1\subseteq\mathcal X$. Let $\mathcal X_2 = \mathcal X\backslash\mathcal X_2$, $|\mathcal X_1| = k$, and $|\mathcal X'_2| = n- k$. After applying a coordinate permutation such that nodes in $\mathcal X_1$ appear first, we have that %the matrices $(A,B;F)$ have the form
\begin{equation}
    A = \begin{bmatrix} A_{11} & 0 \\ A_{21} & A_{22} \end{bmatrix}, \,\,
    B = \begin{bmatrix} 0 \\ B_2 \end{bmatrix}, \,\, 
    F = \begin{bmatrix} F_1 & 0 \\ 0 & F_2 \end{bmatrix},
\end{equation}
where $A_{11}\in\R^{k\times k}$, $A_{22}\in\R^{(n-k)\times (n-k)}$, $B_2\in\R^{(n-k)\times p}$, and other matrices have consistent dimensions. Given Assumption~\ref{assump.independent}, $F_1 \in \R^{r_1\times k}$ correspond to the subset $\mathcal T_1\subseteq\mathcal X_1$, and $F_2\in\R^{(r-r_1)\times (n-k)}$ to $\mathcal T_2 = \mathcal T\backslash \mathcal T_1$. It follows that
\begin{equation}
% \begin{aligned}
    \rank (F\mathcal C) =\rank 
    % \left(
    \begin{bmatrix}
        F_1 & 0 \\ 0 & F_2
    \end{bmatrix}
    \begin{bmatrix}
        0 \\ \mathcal C_2
    \end{bmatrix}
    % \right)
    =
    \rank(F_2\mathcal C_2) 
    % \\
    \leq r-r_1,
% \end{aligned}
\end{equation}
\noindent
where $\mathcal C_2$ is the controllability matrix of $(A_{22},B_2)$. Thus, if there exists a target node with no path coming from a driver node, then $\mathcal T_1\neq \emptyset$, $r_1 > 0$, and condition \eqref{eq.outputctrb} is violated. \QEDwhite

%=======================================================================
% \begin{ack}                               
% \end{ack}

% \bibliographystyle{IEEEtran}
% \bibliography{library}

% Generated by IEEEtran.bst, version: 1.14 (2015/08/26)

\end{document}